\documentclass{article}

\author{Andrzej Hanyga\\
ul. Bitwy Warszawskiej 1920r 14/52\\
02-366 Warszawa, PL}

\title{Plane waves in anisotropic viscoelastic media}

\usepackage{amsmath,amssymb,upgreek,enumerate}

\usepackage{amsthm}
\newcommand{\tens}[1]{\mathsf{#1}}
\newcommand{\e}{\mathrm{e}}
\newcommand{\ee}{\mathbf{e}}
\newcommand{\f}{\mathbf{f}}
\newcommand{\re}{\mathrm{Re}}
\newcommand{\im}{\mathrm{Im}}
\newcommand{\Y}{\mathbf{Y}}
\newcommand{\aA}{\mathbf{a}}
\newcommand{\dd}{\mathrm{d}}
\newcommand{\ii}{\mathrm{i}}
\newcommand{\x}{\mathbf{x}}
\newcommand{\n}{\mathbf{n}}
\renewcommand{\u}{\mathbf{u}}
\newcommand{\vv}{\mathbf{v}}
\newcommand{\w}{\mathbf{w}}
\newcommand{\kk}{\mathbf{k}}
\newcommand{\A}{\mathbf{A}}
\newcommand{\B}{\mathbf{B}}
\newcommand{\C}{\mathbf{C}}

\newcommand{\I}{\mathbf{I}}
\newcommand{\Q}{\mathbf{Q}}
\newcommand{\K}{\mathbf{K}}
\newcommand{\M}{\mathbf{M}}
\newcommand{\N}{\mathbf{N}}

\newcommand{\G}{\mathbf{G}}

\newcommand{\Div}{\mathrm{div}\,}

\newcommand{\OO}{\mathrm{O}}

\newtheorem{theorem}{Theorem}[section]
\newtheorem{lemma}[theorem]{Lemma}
\newtheorem{corollary}[theorem]{Corollary}

\newtheorem{definition}[theorem]{Definition}

\begin{document}

\maketitle

\begin{abstract}
Two concepts of plane waves in anisotropic viscoelastic media are studied. One of these 
concepts allows for the use of methods based on the theory of complete Bernstein functions. 
This allows for a deeper study of frequency-domain asymptotics of the attenuation function and time-domain regularity at the wavefronts. A relation between the direction of the energy flux density 
and the attenuation vector is examined under much more general assumptions.
\end{abstract}

\noindent{\small \textbf{Keywords}: viscoelasticity, anisotropy, completely monotonic, 
causal positive definite, matrix-valued complete \vspace{0.5cm}
Bernstein function, attenuation, dispersion}

\begin{small}
\noindent\textbf{Notation.}\\

\begin{tabular}{lll}
$\mathbb{R}$ & & the set of real numbers \\
$\mathbb{C}$ & & the complex plane \\
$\mathbb{N}$ && the set of positive integers \\
$\mathbb{R}^d, \mathbb{C}^d$ & & real, complex $d$-dimensional space\\
$]a,b]$ & $\{ x \in \mathbb{R} \mid a < x \leq b\}$ & \\
$\mathbb{R}_+$ & $]0,\infty[$ & \\
$\mathcal{M}_d$ & & the space of $d\times d$ real matrices, $d \in \mathbb{N}$\\
$\mathcal{M}^{\mathbb{C}}_d $ & & the space of $d\times d$ complex matrices, $d \in \mathbb{N}$\\
$\mathcal{M}^{\mathbb{C}}$ & $\bigcup_{d\in \mathbb{N}} \mathcal{M}^{\mathbb{C}}_d$\\
$\overline{z}$ & & complex conjugate of $z \in \mathbb{C}$\\
$\Im z$, $\Re z$ & & imaginary, real part of $z$\\
$f\ast_t g$ & $\int_0^\infty f(s) \, g(t-s)\, \dd s$ & Volterra convolution\\
$\tilde{f}(p)$ & $\int_0^\infty \e^{-p t} \, f(t) \, \dd t$ & Laplace transform of $f$\\
$\vv^\top, \A^\top$ & & transpose matrices\\
$\vv^\dag, \A^\dag$ & & Hermitian conjugate matrices\\
$\vv \cdot \w$ & $\vv^\top\, \w$ & scalar product of $\vv, \w \in \mathbb{R}^d$ \\
$\I, \; \I_d$ & & unit matrix, $d\times d$ unit matrix\\
\end{tabular}
\end{small}

\section{Introduction.}

Real viscoelastic media have a universal property which is usually overlooked or not fully exploited.
The relaxation modulus of a real viscoelastic medium is always a completely monotone (CM) function. 
We recall here that an infinitely differentiable function on $\mathbb{R}_+$ is CM if
its derivatives satisfy the infinite sequence of inequalities 
\begin{equation} \label{eq:CM}
(-1)^n \, f^{(n)}(t) \geq 0 \quad \text{ for $t > 0$, $n \in \mathbb{N}\cup \{0\}$}
\end{equation}
In accordance with its physical meaning it is usually bounded but some theoretical models involve unbounded
but locally integrable (LICM) relaxation moduli \cite{Rouse53}. The LICM property of the relaxation modulus and its consequences for creep compliance was studied in
detail in \cite{HanDuality}.
 
A detailed investigation of wave dispersion and attenuation in viscoelastic materials with LICM relaxation modulus in one-dimensional problems can be found in 
\cite{SerHan2010,HanWM2013,HanUno,HanDue,HanJCA}. It was shown in \cite{SerHan2010,HanWM2013} that the wavenumber vector multiplied by 
the imaginary unit is a complete Bernstein function \cite{BernsteinFunctions} of the variable $p := -\ii \omega$, where $\omega$ denotes the circular frequency. Integral 
representations of the attenuation function (the logarithmic attenuation as a function of frequency)
and the inverse phase speed were thus derived. The integral representation allowed asymptotic analysis of the attenuation and regularity of the wave field at the wavefront.

One-dimensional problems do not exhaust all the practical applications of viscoelasticity. 
Anisotropy is an important property of rocks, bones and bio-tissues. It is often studied 
simultaneously with viscous properties of materials. It however turns out that the three-dimensional solutions of the anisotropic viscoelastic equations of motion cannot be represented by the same 
classes of functions as in one-dimensional problems. It was shown in \cite{HanAnisoVE} that 
Green's functions in such media have to be expressed in terms of matrix-valued complete Bernstein 
functions (mvCBFs). The last concept was developed in \cite{HanAnisoVE}. 

It is common to
begin investigation of wave propagation with a detailed study of plane waves, their dispersion, attenuation and polarization. This paper is an attempt to review the plane wave concept in
the context of anisotropy and viscoelasticity using the methods of \cite{HanAnisoVE}. 
This paper is dedicated to a few selected problems and its objective is to demonstrate 
the appropriate mathematical tools.

Combination of anisotropy and viscosity leads to some ambiguities in 
the definition of plane waves. One would expect that plane waves represent independent 
propagation modes. This approach however excludes an application of the full power of the theory
based on LICM relaxation and the mathematical apparatus of complete Bernstein functions. In order to retain the possibility of using the complete Bernstein functions it is 
necessary to work with coupled plane wave modes (Section~\ref{sec:alter}). 

An additional topic investigated here is the relation between the attenuation vector (the
imaginary part of the complex wave number vector) and the energy flux density vector. For
this topic we have chosen a larger class of relaxation moduli, namely 
(matrix-valued) causal positive definite (CPD) functions. While the assumption that a relaxation modulus 
completely monotone is a guess based on experimental data, the CPD property can be
derived from the fluctuation-dissipation theorem using the Wiener-Khintchine theorem. More importantly, 
the theory CPD function is more suitable for the problem at hand.

The paper contains a brief presentation of 
matrix-valued complete Bernstein functions (Section~\ref{sec:mvCBF}). In the following 
two sections complete Bernstein functions are used to discuss two concepts of plane waves in 
anisotropic viscoelastic media. 

In Sec~\ref{sec:CPD} we define the causal positive definite functions and use their properties to prove that in a viscoelastic 
medium with a causal positive definite relaxation modulus the angle between the energy flux 
density and the attenuation vector is acute. This statement is a fortiori true for media with
LICM relaxation moduli.

\section{Matrix-valued complete Bernstein functions.}
\label{sec:mvCBF}

\begin{definition} \label{def:mvCBF}
A matrix-valued function $\A: \mathbb{R}_+ \cup \{ 0 \} \rightarrow \mathcal{M}_d$ 
is said to be a matrix-valued complete Bernstein function (mvCBF) if 
it has an analytic continuation to the cut complex plane $\mathbb{C}\setminus\, ]-\infty,0]$ 
which satisfies the inequality
\begin{equation} \label{eq:0a} 
\Im z \, \Im \A(z) \geq 0
\end{equation}
and $\lim_{x \rightarrow 0+} \A(x)$ exists and is real.
\end{definition}
By definition $\Im \A(z)$ is Hermitian, hence the inequality in \eqref{eq:0a} 
makes sense. 

\begin{theorem} \label{thm:CBF}
If the matrix-valued function $\A(x)$ on $\mathbb{R}_+$ is a mvCBF then there are two positive semi-definite 
matrices $\B$ and $\C$, a positive Radon measure $\mu$ on $\mathbb{R}_+$ satisfying the
inequality
\begin{equation} \label{eq:6}
\int_{]0,\infty[} (1 + r)^{-1} \,\mu(\dd r) < \infty,
\end{equation}
and a measurable $\mu$-almost everywhere bounded and positive semi-definite function $\mathbf{M}$
such that 
\begin{equation} \label{eq:5}
\A(x) = \B + x\, \C + x \int_{]0,\infty[} (x + s)^{-1} \mathbf{M}(s)\, \mu(\dd s), \qquad x \geq 0
\end{equation}
\end{theorem}
\begin{proof}
If $\B \in \mathcal{M}^\mathbb{C}_d$ and $\Im \B \geq 0$ then for every $\vv \in
\mathbb{C}^d$ 
$$\Im \vv^\dag\, \B \, \vv = \frac{1}{2 \ii} \left[ \vv^\dag \, \B \, \vv - 
\overline{\vv^\dag \, \B \, \vv}\right] = \frac{1}{2 \ii} \vv^\dag\, \left(\B - \B^\dag\right)\, 
\vv \geq 0$$

If $\A$ is a mvCBF then 
$\Im \left[\vv^\dag\, \A(z)\, \vv\right] \geq 0$ for every $\vv \in \mathbb{C}^d$ and $\Im z > 0$.
Furthermore $\lim_{x\rightarrow 0+} \A(x)$ exists and is real. 
For real $x$ the matrix $\A(x)$ is real, hence it is sufficient to consider $\vv \in \mathbb{R}^d$.
It follows that  $\vv^\top\, \A(x)\, \vv$ is a CBF for every $\vv \in \mathbb{R}$ and 
therefore there are two non-negative numbers $b_\vv, c_\vv$ and a positive measure $\mu_\vv$ satisfying 
inequality~\eqref{eq:6} such that 
$$ \vv^\top\, \A(x)\, \vv = b_\vv + x \, c_\vv + x \int_{]0,\infty[} (x + r)^{-1} \, \mu_\vv(\dd r)$$
\cite{BernsteinFunctions}. It is clear that $b_\vv = \vv^\top \, \B\, \vv$, where $\B := \A(0)$.
As for $c_\vv$,
$$c_\vv = \lim_{x\rightarrow\infty} \left[ x^{-1}\vv^\top\, \A(x)\, \vv\right]$$
This proves that the limit on the right exists, hence a symmetric matrix $\C$ can be defined by polarization
$$2 \vv^\top \, \C \, \w = c_{\vv + \w}- c_\vv - c_\w $$
It is easy to see that the right-hand side is a linear function of $\vv$ and $\w$ and therefore it defines a symmetric positive semi-definite matrix $\C$. 

Let $\A_0(x) := \A(x) - \B - x \, \C$,
$$\vv^\top\, \A_0(x)\, \vv = x \int_{]0,\infty[} (x + s)^{-1}\, \mu_\vv(\dd s)$$
By polarization 
$$\vv^\top\, \A_0(x)\, \w = \int_{]0,\infty[} (x + s)^{-1}\, N_{\vv,\w}(\dd r)$$
where $2 N_{\vv,\w}(I) := \mu_{\vv+\w}(I) - \mu_\vv(I) - \mu_\w(I)$ for every interval 
$I \subset \mathbb{R}_+$. $N_{\vv,\w}(I)$ can be expressed in terms of
$\vv^\top\,\A_0\, \w$:
$$N_{\vv,\w}(I) = \lim_{\varepsilon\rightarrow 0+} \frac{1}{\uppi} \int_I \Im 
\frac{\vv^\top\, \A_0(-s + \ii \varepsilon) \, \w}{s - \ii \varepsilon} \, \dd r$$
for every segment $I$ whose ends are continuity points of the measure $N_{\vv,\w}$
\cite{BernsteinFunctions}. 
The last expression shows existence of a $\mathcal{M}_d$-valued measure $\N$ such that 
$N_{\vv,\w}(I) = \vv^\top\, \N(I)\, \w$. $\N(I)$ is positive semi-definite because
for every $\vv \in \mathbb{R}_+$ $\vv^\top\, \N(I)\, \vv \geq 0$.

Now 
$$\vv^\top\, \N(I) \, \vv + \w^\top \, \N(I)\, \w \pm 2 \vv^\top\, \N(I)\, \w
= (\vv \pm \w)^\top\, \N(I)\, (\vv \pm \w) \geq 0$$
hence
\begin{equation} \label{eq:RN}
\vert \vv^\top\, \N(I) \, \w  \vert \leq \frac{1}{2} \left[ \vv^\top\, \N(I)\, \vv
+ \w^\top \, \N(I) \, \w\right] \leq \frac{1}{2} \left(\| \vv \|^2 +  \, \| \w\|^2\right) \, \mu(I)
\end{equation}
where $\mu(I) := \mathrm{trace}[\N(I)]$. The Radon measure $\mu$ thus defined on
$\mathbb{R}_+$ is positive. The measures $\mu_\vv$ satisfy the inequality~\eqref{eq:6},
hence $\N_\vv$ and $\mu$ satisfy the same inequality.

By the Radon-Nikodym theorem \cite{Rudin76} there is a $\mu$-almost everywhere bounded function
$\M(r)$ on $\mathbb{R}_+$ such that $\N(\dd r) = \M(r)\, \mu(\dd r)$. 
\end{proof}

\begin{corollary} \label{cor:calc}
If $\A(x)$ is a mvCBF satisfying~\eqref{eq:5}
$$\N(]a,b]) = \frac{1}{\uppi} \lim_{\varepsilon\rightarrow 0+}\int_{]a,b]} 
\Im \left[\frac{\A_0(-s + \ii \varepsilon)}{s - \ii \varepsilon}\right] \dd s$$
and $\mu(]a,b]) = \mathrm{trace}[\N(]a,b])]$
for every regular point $b > 0$ of the measure $\N$ and every $a > 0$,
where $\A_0(x) = \A(x) - \B - x \, \C$. 
\end{corollary}

Equation~\eqref{eq:5} provides an analytic continuation of the function $\A(x)$ to the 
complex plane cut along the negative real semi-axis. 

The function $\M(s)$ can be assumed bounded by 1 almost everywhere in the sense of measure $\mu$.

\begin{corollary} \label{cor:coeffs}
If $\A(x)$ satisfies equation~\eqref{eq:5} and \eqref{eq:6}, then
\begin{enumerate}[(i)]
\item $\A(0) = \B$;
\item $\C = \lim_{x\rightarrow \infty} x^{-1}\, \A(x) \quad \text{for $x \in \mathbb{R}_+$}$.
\end{enumerate}
\end{corollary}
\begin{proof}
(i) follows from equation~\eqref{eq:5};

Concerning (ii), it is sufficient to prove that
$\mathbf{L}(x) := \int_{]0,\infty[} (x + s)^{-1} \M_\n(s) \, \mu(\dd s) \rightarrow 0$
for $x \rightarrow \infty$.
Indeed, for $ x \geq 1$,
$$\vert \mathbf{L}(x) \vert \leq \int_{]0,\infty[} (x + s)^{-1}\, \mu(\dd s) \leq 
\int_{]0,\infty[} (1 + s)^{-1}\, \mu(\dd s) < \infty$$
on account of \eqref{eq:6}. The Lebesgue Dominated Convergence Theorem implies that 
$\lim_{x\rightarrow\infty} \mathbf{L}(x) = 0$, q.e.d. 
\end{proof} 

\begin{lemma}
\begin{equation} \label{eq:a1}
x^\alpha = \frac{\sin(\alpha\,\uppi)}{\uppi} \int_0^\infty \frac{x \, s^{\alpha-1}}{x + s} 
\dd s \qquad \text{for $x \geq 0$}
\end{equation}
\end{lemma}
\begin{proof}
Identity~\eqref{eq:a1} is equivalent to the identity
\begin{equation} \label{eq:a2}
x^{\alpha-1} = \frac{\sin(\alpha\,\uppi)}{\uppi} \int_0^\infty \frac{s^{\alpha-1}}{x + s} 
\dd s \qquad\text{ for $x \geq 0$}
\end{equation} 

In order to prove the last identity we use the fact that the Stieltjes transform is an iterated 
Laplace transform and
$$\int_0^\infty s^{\alpha-1} \, \e^{-y s}\, \dd s = y^{\alpha-1} \, \int_0^\infty 
z^{\alpha-1}\, \e^{-z} \, \dd z = \Gamma(\alpha)\, y^{\alpha-1}$$
$$\Gamma(\alpha) \int_0^\infty y^{-\alpha} \, \e^{-x y} \, \dd y = \Gamma(\alpha)\,
\Gamma(1-\alpha) \, x^{\alpha-1} = \frac{\uppi}{\sin(\alpha\, \uppi)} x^{\alpha-1}$$ 
which proves \eqref{eq:a2}. 
\end{proof} 

This suggests an integral definition of the square root $\B^{1/2}$:
\begin{equation} \label{eq:a3}
\B^{1/2} = \uppi \int_0^\infty \B\, (s\,\I + \B)^{-1}\, 
s^{-1/2}\, \dd s
\end{equation}
provided that $\B$ has no non-positive real eigenvalue.

If $f$ is a CBF then there are two non-negative real numbers $a, b$ and a positive Radon 
measure $\rho$ on $\mathbb{R}_+$ such that
$$f(x) = a + b \, x + \int_{]0,\infty[} \frac{x}{s + x} \rho(\dd s)$$
Using this formula the function $f$ can be extended to $\mathcal{M}^\mathbb{C}$:
\begin{equation} \label{eq:defbyCBF}
f(\B) := a \, \I + b \, \B + \int_{]0,\infty[} \B\, (s \, \I + \B)^{-1} \, \rho(\dd s)
\end{equation}

The function $x^{1/2}$ is a CBF and equation~\eqref{eq:a3} is a particular case of \eqref{eq:defbyCBF}.  

In order to prove that the square root $\B^{1/2}$ is a square root in the algebraic sense 
we shall use the following theorem \cite{Schilling11}, Theorem~4.1 (5):
\begin{theorem}
If $f, g$ and their pointwise product $f\, g$ are CBFs then 
$(f\, g)(\B) = f(\B)\, g(\B)$ for every matrix $\B \in \mathcal{M}^\mathbb{C}$, 
where $f(\B)$ and $g(\B)$ are defined by equation~\eqref{eq:defbyCBF}.
\end{theorem}
The functions $f(x) = g(x) = x^{1/2}$ and $f\, g$ are CBFs, 
and the definition \eqref{eq:a3} has the form of equation~\eqref{eq:defbyCBF},
hence
\begin{equation} \label{eq:sqroot} 
\B^{1/2} \, \B^{1/2} = \B
\end{equation}

The two square roots of a matrix $\B$ in the algebraic sense are defined as solutions $\mathbf{Y}$ of the equation
\begin{equation} \label{eq:rooteq}
\mathbf{Y}^2 - \B = 0
\end{equation}
If a matrix $\B \in \mathcal{M}_d^\mathbb{C}$ has no non-positive real eigenvalue
then equation~\eqref{eq:rooteq}  has a unique solution with the property that its eigenvalues lie in the
open right half of the complex plane \cite{Higham87}. This particular solution is called the 
principal square root.

\begin{lemma} \label{lem:sqr}
If $\Re \B > 0$  
then the square root $\mathbf{Y} = \B^{1/2}$ defined by equation~\eqref{eq:a3} is the principal 
square root. 
\end{lemma}
\begin{proof}
Equation~\eqref{eq:sqroot} implies that $\mathbf{Y}$ satisfies equation~\eqref{eq:rooteq}. The real part of the integrand of \eqref{eq:a3} can be expressed in the form 
\begin{multline*}
\frac{1}{2} \left[\B\, \left( \B + s\, \I\right)^{-1} +  \left(\B^\dag + s \, \I\right)^{-1} \, \B^\dag\right] =\\
\frac{1}{2} \left(\B^\dag + s \, \I\right)^{-1}\, \left[ \B^\dag \, (\B + s\, \I) + \left(\B^\dag + s \, \I\right) \, \B\right]\,
(\B + s\, \I)^{-1} =  \mathbf{U}^\dag \,\left[\B^\dag\, \B + s\, \Re \B\right]\,\mathbf{U}
\end{multline*}
where $\mathbf{U} := (\B + s \, \I)^{-1}$. Hence for $s > 0$ the real part of the integrand of \eqref{eq:a3} is positive
 and therefore the square root defined by \eqref{eq:a3} satisfies the inequality $\Re \B^{1/2} > 0$.
This result along with equation~\eqref{eq:sqroot} implies that the right-hand side of \eqref{eq:a3}
is the principal square root.  
\end{proof} 

\begin{theorem} \label{thm:a1}
If $\B \in \mathcal{M}^\mathbb{C}_d$ and $\Im \B \geq 0$ then the square root defined by \eqref{eq:a3} satisfies the inequality 
$\Im \B^{1/2} \geq 0$. 
\end{theorem} 
\begin{proof}
\begin{multline*}
\frac{1}{2 \ii} \left[ \B \, (\B + s \I)^{-1} - \left(\B^\dag + s\, \I\right)^{-1}\, \B^\dag\right] = \\
\frac{1}{2 \ii} \left(s\, \I + \B^\dag \right)^{-1}\, \left[\left(s \, \I + \B^\dag\right)\, \B
 - \B^\dag (\B + s\, \I)\right] \, (\B + s\, \I)^{-1} = 
s \, \mathbf{U}^\dag\, \Im \B \, \mathbf{U} \geq 0 
\end{multline*}
where $\mathbf{U} := (s\, \I + \B)^{-1}$. In view of equation~\eqref{eq:a3} this implies that 
$\Im \B^{1/2} \geq 0$. 
\end{proof}

In view of Definition~\ref{def:mvCBF} this entails the following important corollary:
\begin{corollary}\label{cor:half}
If $\A(x)$ is a mvCBF then the principal square root $\A(x)^{1/2}$ is a mvCBF.
\end{corollary}

\begin{definition} \label{def:Stieltjes}
A matrix-valued function $\A: \mathbb{R}_+ \cup \{ 0 \} \rightarrow \mathcal{M}_d$ 
is said to be a Stieltjes function if 
it has an analytic continuation to $\mathbb{C}\setminus\, ]-\infty,0]$ 
which satisfies the inequality
\begin{equation} \label{eq:1a} 
\Im z \, \Im \A(z) \leq 0
\end{equation}
and $\lim_{x \rightarrow 0+} \A(x)$ exists and is real positive semi-definite.
\end{definition}

If $\B \in \mathcal{M}^\mathbb{C}_d$ is invertible and $\Im \B \geq 0$ then
$$\Im \B^{-1} = \frac{1}{2 \uppi \ii} \left(\B^{-1} - \B^{\dag -1}\right) =
\frac{1}{2 \uppi \ii} \B^{-1}\, \left(\B^\dag - \B\right)\, \B^{\dag -1} \leq 0$$
This proves the following lemma:
\begin{lemma} \label{lem:inverse}
If $\A(x)$ is a mvCBF and $\A(z)$ is invertible for $z \not\in \,]-\infty,0]$
then $\A(x)^{-1}$ is a matrix-valued Stieltjes function.
\end{lemma}

\begin{theorem}\label{thm:S}
If the matrix-valued function $\A(x)$ is a Stieltjes function then there are two positive 
semi-definite 
matrices $\B$ and $\C$, a positive Radon measure $\mu$ on $\mathbb{R}_+$ satisfying the
inequality~\eqref{eq:6} and a
measurable $\mu$-almost everywhere bounded and positive semi-definite matrix-valued function 
$\mathbf{M}$ such that 
\begin{equation} \label{eq:S}
\A(x) = \B + x^{-1}\, \C + \int_{]0,\infty[} (x + s)^{-1} \mathbf{M}(s)\, \mu(\dd s)
\end{equation}
\end{theorem}

The proof of this theorem, based on the integral representation 
$$f(x) = a + \frac{b}{x} + \int_{]0,\infty[} (x + s)^{-1} \, \mu(\dd s)$$
\cite{BernsteinFunctions}
with $a, b \geq 0$ and a positive Radon measure $\mu$ satisfying equation~\eqref{eq:6},
is analogous to the proof of Theorem~\ref{thm:CBF}. 

Equation~\eqref{eq:S} provides an analytic continuation of the function $\A(x)$ to the complex plane cut along the negative real semi axis. 

Comparison of Theorems~\ref{thm:CBF} and \ref{thm:S} yields the following corollary
\begin{corollary} \label{cor:toS}
If $\A$ is a mvCBF then $x^{-1}\, \A(x)$ is a matrix-valued Stieltjes
function.
\end{corollary}

\section{Plane waves in an anisotropic viscoelastic medium.}
\label{sec:plwav}

We shall consider complex plane wave solutions 
\begin{equation} \label{eq:plwavc}
\exp\left(-\ii \omega t + \ii \kk\cdot \x\right)\, \aA 
\end{equation}
with $\kk, \aA \in \mathbb{C}^3$,  of the homogeneous equation of motion 
\begin{equation} \label{eq:homoeqmotion}
\rho \, \u_{,tt} = \nabla \cdot[\tens{G}(t)\ast \nabla \u_{,t}]  \qquad t > 0,\quad
\end{equation}
where $\tens{G}(t)\ast \nabla \u_{,t}$ is $G_{ijrs}\ast \partial u_{r,s}/\partial t$ 
componentwise. 

Substituting \eqref{eq:plwavc} in \eqref{eq:homoeqmotion} we obtain an eigenproblem
\begin{equation} \label{eq:eigvone}
\left[\kappa^2 \, Q_{ijrs}(-\ii \omega) n_j\, n_s + \rho\, \omega^2\, \delta_{ir}\right] a_r = 0
\end{equation}
where we have set $\kk = -\ii \,\kappa\, \n$ with $\n \in \mathbb{R}^3$, $\n^2 = 1$ and 
\begin{equation}
\tens{Q}(p) := p \, \tilde{\tens{G}}(p)
\end{equation}

Let $\Q_\n$ denote the rank-2 tensor $Q_{ijrs}(-\ii \omega) n_j\, n_s$. We shall assume
the strong ellipticity condition\footnote{In viscoelastic terminology this condition 
means that the medium is a viscoelastic solid.} 
$$G_{klrs}^{\infty}\, n_l\, n_s\, a_k\, a_r > 0 \quad \text{for non-zero vectors $\aA, \n$}$$
where $\tens{G}^\infty = \lim_{t\rightarrow \infty} \tens{G}(t) = \lim_{p\rightarrow 0} \tens{Q}(p)$.
Let $\tens{G}^\infty$ denote the matrix $G^\infty_{klrs} \, n_l \, n_s$. 
The matrix-valued function $\Q_\n(p)$ is a CBF, hence it is non-decreasing and therefore
$$\Q_\n(p) \geq \Q_\n^\infty := \tens{G}^\infty_\n > 0$$
and the matrix $\Q_\n$ is invertible. Consequently we can define the matrix-valued function
\begin{equation} \label{eq:Kndef}
\K_\n(p) := \rho \, p\,\Q_\n(p)^{-1/2}
\end{equation}
and recast equation~\eqref{eq:eigvone} in a simpler form
\begin{equation}  \label{eq:eigvtwo}
\left[\K_\n(-\ii \omega)^2 - \kappa^2 \, \I \right] \, \aA = 0
\end{equation} 

A necessary condition for the existence of a non-trivial solution $\aA$ of \eqref{eq:eigvtwo} is
the dispersion equation:
\index{dispersion equation}
\begin{equation} \label{eq:deteq}
D := \det\left[\K_\n(-\ii \omega)^2 - \kappa^2 \, \I \right] = 0
\end{equation}

Equation~\eqref{eq:eigvtwo} is satisfied if either 
$\K_\n(-\ii \omega) \, \aA = \kappa \, \aA$ or
$\K_\n(-\ii \omega) \, \aA = -\kappa \, \aA$. This alternative corresponds to two different propagation directions $\pm \n$. 

Equation~\eqref{eq:deteq} is an algebraic equation of degree 3 in $\lambda := \kappa^2$, hence it has three roots $\kappa^2 = \lambda_j(\omega,\n)$, $j = 1,2,3$, which are
eigenvalues of $\K_\n(-\ii \omega)^2$. 
For each root 
$\kappa^2 = \lambda_j(\omega,\n)$, $j \in \{ 1, 2, 3\}$ equation~\eqref{eq:eigvtwo} has at least one non-trivial solution 
$\aA \in \mathbb{C}^3$. This solution can be obtained from the expansion of the determinant $D$ in terms 
of the elements of a row and its $2\times2$ minors.
Solutions $\aA$ corresponding to numerically different roots $\lambda_j(\omega,\n)$ 
are linearly independent.
If a root $\lambda_j(\omega,\n)$ is double, e.g. $\lambda_2(\omega,\n) =  \lambda_3(\omega,\n)$
for the given values of $\omega$ and $\n$ then the number of linearly independent 
solutions corresponding to the double eigenvalue can be either 1 or 2. Summarizing 
for each solution of the dispersion equation the equation of motion \eqref{eq:homoeqmotion} has at least one non-trivial complex solution \eqref{eq:plwavc} 
with $\kk = \pm \ii \lambda_j(\omega)^{1/2} \, \n$ and $\A$ satisfying equation~\eqref{eq:eigvtwo}.
The complex conjugate of \eqref{eq:plwavc} is also a solution of \eqref{eq:homoeqmotion},
hence \eqref{eq:homoeqmotion} has a non-trivial plane wave solution
\begin{equation} \label{eq:plwav}
\u(t,\x) = \Re \left[ \exp\left(-\ii \omega t + \ii \kk\cdot \x\right)\, \aA \right]
\end{equation}
for each root $\lambda_j(\omega,\n)$ of $D$.

\begin{theorem} \label{thm:K}
$\K_\n(\cdot)$ is a mvCBF for every $\n \in \mathcal{S}$.
\end{theorem}
\begin{proof}
$\Q_\n(p)$ is a mvCBF, hence $\Q_\n(p)^{1/2}$ is a mvCBF (Corollary~\ref{cor:half}). By 
Corollary~\ref{cor:toS}
$p^{-1}\, \Q_\n(p)^{1/2}$ is a matrix-valued Stieltjes function.
By Lemma~\ref{lem:inverse}  its inverse $p\, \Q_\n(p)^{-1/2}$ is a CBF. 
\end{proof} 

Since $\K_\n(0) = 0$, Theorem~\ref{thm:CBF} implies that 
\begin{equation} \label{eq:6a}
\K_\n(p) = p \, \B_\n + p \int_{]0,\infty[} (p + r)^{-1} \, \M_\n(r)\, \mu(\dd r)
\end{equation}
where $\mu$ is a positive Radon measure satisfying the inequality \eqref{eq:6},
$\M_\n(r)$ is an $S$-valued function bounded $\mu$-almost everywhere by 1
and $\B_\n \in S$ is positive semi-definite. The subscript $\n$ indicates the parametric
dependence on $\n$. 
Corollary~\ref{cor:coeffs}  implies that 
$$\B_\n = \lim_{p \rightarrow \infty} \left[p^{-1} \, \K_\n(p)\right] = 
\rho^{1/2}\, \left[\lim_{p \rightarrow \infty} \Q_\n(p)\right]^{-1/2} =
\left[\rho^{-1} \,\G^0_\n\right]^{-1/2}$$
where $p$ is considered a real variable, ${G^0_{\n}}_{\,ij} := G^0_{\;ikjl} \, n_k\, n_l$,
is a real symmetric matrix. It can be diagonalized and its eigenvalues 
are non-negative. A non-zero eigenvalue $1/c^{(0)}_\n$ of $\B_\n$ represents 
the inverse wavefront speed for the corresponding mode and a plane wave front
orthogonal to the vector $\n$. 

Equation~\eqref{eq:6a} implies that the matrix-valued function $\K_\n(p)$ can be 
extended to $\mathbb{C}\setminus [0,\infty[$ by analytic continuation and
\begin{equation} \label{eq:6b}
\K_\n(-\ii \omega) = -\ii \omega \C_\n(\omega) + \A_\n(\omega)
\end{equation}
with
\begin{gather}
\C_\n(\omega) := \B + \int_{]0,\infty]} \frac{r \, \M_\n(r)}{r^2 + \omega^2} 
\mu(\dd r) \label{eq:frDisp}\\
\A_\n(\omega) := \omega^2 \int_{]0,\infty]} \frac{\M_\n(r)}{r^2 + \omega^2} 
\mu(\dd r) \label{eq:frAtt}
\end{gather}
Both matrix-valued functions are symmetric and positive semi-definite for $\omega \in \mathbb{R}$. 
Note that $\C_\n(\omega)$ is a non-increasing function of $\omega$ with respect to the 
usual order relation for matrices and its tends to the limit $\B$
for $\omega \rightarrow \infty$. This implies that the eigenvalues of the matrix-valued 
function $\C_\n(\omega)$ are non-increasing functions of $\omega$ bounded from below by the
eigenvalues of $\B$. 
The matrix-valued function $\A_\n(\omega)$ is non-decreasing and
it is bounded from above if $\mu$ has a finite mass. 

If $\vv(\omega)$ is an 
eigenvector of $\K_\n(-\ii \omega)^2$ corresponding to the eigenvalue $\lambda_j(\omega,\n)$ and
satisfying the normalization condition  $\vv^\dag\, \vv = 1$ then 
$a_j(\omega,\n) := \Re \lambda_j(\omega,\n) = \vv(\omega)^\dag\, \A_\n(\omega) \vv(\omega) \geq 0$ and 
$1/c_j(\omega,\n) := -\Im \lambda_j(\omega,\n)/\omega = \vv(\omega)^\dag\, \C_\n(\omega) \vv(\omega)  \geq 0$. 
Thus each root $\lambda_j$ of the dispersion equation has the form 
$\lambda_j(\omega,\n) = -\ii \omega/c_j(\omega,\n) + a_j(\omega,\n)$ and the plane wave is given by the formula
\begin{equation}
\u(t,\x) = \Re \left[ \e^{-\ii \omega t + \ii \omega \, \n\cdot\x /c_j(\omega,\n) - a_j(\omega,\n)\, \n\cdot \x}\, \aA \right]
\end{equation}
where the amplitude $\aA = \gamma \vv $, $\gamma \in \mathbb{C}$,  is a solution of 
equation~\eqref{eq:eigvtwo} for $\kappa = \lambda_j$.

\section{Plane waves in anisotropic viscoelastic media: an alternative approach.}
\label{sec:alter}

The function $\lambda_j(\cdot,\n)$ is not in general a CBF and therefore the functions $1/c_j(\cdot,\n)$ and $a_j(\cdot,\n)$ do not have integral representations analogous to \eqref{eq:frDisp} 
and \eqref{eq:frAtt}. We shall now try to work around this problem.

\begin{theorem} \label{thm:consteigv}
If $\vv$ is a constant eigenvector of a mvCBF $\Y(y)$ then the corresponding eigenvalue is
a CBF.
\end{theorem}
\begin{proof}
Theorem~\ref{thm:CBF} implies that 
$$\Y(y) = \A + y \B + y \int_{]0,\infty[} (y + s)^{-1} \mathbf{M}(s) \, \mu(\dd s)$$
where $\A$, $\B$ are positive semi-definite matrices, $\mu$ is a positive Radon measure on
$\mathbb{R}_+$ satisfying inequality \eqref{eq:6} and $\mathbf{M}$ is positive semi-definite
$\mu$-almost everywhere.
Let $\vv$ be a unit vector. Then the eigenvalue
$$\lambda(y) = \vv^\dag\, \Y(y) \vv = a + y\, b  
+ y \int_{]0,\infty[} (y + s)^{-1}  \, m(s)\, \mu(\dd s)$$ 
where $a := \vv^\dag\, \A \vv \geq 0$, $b := \vv^\dag\, \B \vv \geq 0$ 
and $m(s) := \vv^\dag \,\mathbf{M}(s) \,\vv \geq 0$ $\mu$-almost everywhere.
Hence $\lambda(y)$ is a CBF.
\end{proof}
\begin{corollary} \label{cor:prec} 
If $\K_\n(p)$ has a constant eigenvector $\aA$ and $\kappa(p,\n)$ is the corresponding 
eigenvalue then $\kappa(\cdot,\n)$ is a CBF and 
$\kappa(-\ii \omega,\n) = -\ii \omega/c(\omega,\n) + a(\omega,\n)$
where $1/c(\omega,\n), a(\omega,\n) \geq 0$.
\end{corollary}
It is clear that $1/c(\omega,\n)$ is finite, i.e. $c(\omega,\n) > 0$.
\begin{corollary}
Under the hypotheses of Corollary~\ref{cor:prec} 
the functions $1/c(\omega,\n)$ and $a(\omega,\n)$ satisfy equations~\eqref{eq:frDisp}
and \eqref{eq:frAtt} respectively. 
\end{corollary}
The function $c(\omega,\n)$ tends to $c^\infty(\n) := 1/\vv^\dag \, \B \vv$ as 
$\omega \rightarrow \infty$. 

The hypothesis of Theorem~\ref{thm:consteigv} is satisfied in for all the plane wave modes in isotropic viscoelastic media and for longitudinal waves when the wavefront normal $\n$ is parallel to
a specific direction. In the other cases the eigenvalue $\kappa(p,\n)$ is not a CBF. It is therefore 
more convenient to define plane waves in a different way:
\begin{equation} \label{eq:altplane}
\u(t,\x) = \Re \left[ \e^{-\ii \omega t - \K_\n(-\ii \omega) \, \n\cdot \x} \, \aA\right]
\end{equation}
where $\K_\n$ s defined by equation~\eqref{eq:Kndef} and $\aA \in \mathbb{C}^3$ is an {\em arbitrary} vector. It is easy to verify that \eqref{eq:altplane} satisfies equation~\eqref{eq:homoeqmotion}. 
Such plane-wave solutions appear in the analysis of Green's functions in \cite{HanAnisoVE}.

We now take advantage of the fact that $\K_\n(p)$ is a mvCBF and use equation~\eqref{eq:6b}
which yields the formula 
\begin{equation}
\u(t,\x) = \Re \left[\e^{-\ii \omega [t - \C_\n(\omega) \, \n\cdot\x] - \A_\n(\omega)\, \n\cdot \x}\, \aA \right]
\end{equation}
The matrix $\C_\n(\omega)$ plays the role of the inverse phase speed $\mathcal{D}(\omega,\n)$ 
while $\A_\n$ plays the role of the attenuation function $\mathcal{A}(\omega,\n)$.

We now prove that $\u(t,\x)$ is exponentially attenuated in the direction of its propagation.
Define the matrix norm
\begin{equation}
\| \A \|_d := \sup_{\w \in \mathbb{C}^d} \sqrt{\frac{\w^\dag\, \A^\dag\, \A \w}{\w^\dag \w}}
\end{equation} 
The norm $\| \A \|_d$ is unitarily invariant, i. e. for every unitary matrix $\mathbf{U}$
and $\mathbf{V}$ 
$$ \| \mathbf{U} \, \A\, \mathbf{V}\|_d = \| \A \|_d.$$
We can now show that the function 
\begin{equation} \label{eq:Zz} 
\e^{-y \K_\n(\omega)} \equiv \e^{\ii\,\omega\, y \,\mathbf{C}_\n(\omega) - y\, \A_\n(\omega)}
\end{equation}
is exponentially attenuated if $\A_\n(\omega) > 0$ for $\omega \in \mathbb{R}_+$.

According to Theorem~IX.3.11 in \cite{Bhatia97} every complex $d \times d$ matrix $\A$
satisfies the inequality 
$$\| \e^{\A} \|_d \leq \| \e^{\Re \A}\|_d$$
Hence 
$$\| \e^{-y \, \K_\n(-\ii \omega)} \|_3 \leq \| \e^{-y \, \A_\n(\omega)}\|_3$$
Let $a_0(\omega)$ be the smallest eigenvalue of the real symmetric matrix $\A_\n(\omega)$.
If $a_0$ is the smallest eigenvalue of a real symmetric
$d \times d$ matrix $\A$ with eigenvalues $a_j$ and unitary eigenvectors $\vv_j$, then 
$$\| \e^{-\A}\|_d = \sup_{\{c_j\mid j=1,\ldots d\}} \sqrt{\frac{\sum_{j=1}^d \vert c_j\vert^2 \, \exp(-2 a_j)}{\sum_{j=1}^d \vert c_j\vert^2}} = \e^{-a_0}$$
Consequently
\begin{equation} \label{eq:attn}
\| \e^{-y \, \K_\n(-\ii \omega)} \|_3 \leq \e^{-a_0(\omega) y}
\end{equation}
If $\A_\n(\omega) > 0$ for $\omega \geq 0$ then equation~\eqref{eq:attn} implies that
the plane wave with the wavefront normal $\n$ is exponentially attenuated with
distance. 

For sufficiently small anisotropic dispersion $\C_\n(\omega) - w(\omega) \, \I_3 = \OO[\varepsilon]$ 
and attenuation $\A_\n(\omega) = z(\omega) \OO[\varepsilon]$, where $w(\omega)$ and $z(\omega)$ are some real functions, the inverse phase function and the attenuation function in \eqref{eq:Zz} can be approximately disentangled in the form 
\begin{multline} \label{eq:OUN}
\exp(\ii \omega \C_\n(\omega)\, y - \A_\n(\omega) \, y) = \\
\Phi(y, \omega) \,\exp(\ii \omega\, y \C_\n(\omega)) \, \exp(-y\,\A_\n(\omega))
= \\ \Psi(y, \omega)\, \exp(-y \,\A_\n(\omega)) \, \exp(\ii \omega \, y \C_\n(\omega))
\end{multline}
where the functions $\Phi(y, \omega), \Psi(y, \omega) = 1 + \OO\left[y^2\right]$ can be expressed in terms of exponential functions of nested commutators of $\C_\n(\omega)$
and $\A_\n(\omega)$ by the use of the Zassenhaus formula \cite{Magnus54}.

The infinite product is convergent if $\vert y \vert \,
\| \omega\, \C_\n(\omega) + \A_\n(\omega)\| \leq 0.596705$. 
The matrix $\A_\n(\omega)$ is symmetric 
hence it admits a spectral representation with eigenvalues $a_i(\omega),\n$ and eigenvectors $\w_i(\omega,\n)$. Logarithmic attenuation rates 
can be made explicit by applying the second line of equation~\eqref{eq:OUN} and expanding the amplitude vector in terms of the eigenvectors $\w_i$, 
 $\aA = \sum_{i=1}^3 v_i\, \w_i(\omega,\n)$,
 so that 
$$\exp(y [\ii \omega \C_\n(\omega) - \A_\n(\omega)]) \, \aA = \Phi(y, \omega)\, \exp(\ii \omega r 
\C_\n(\omega)) \sum_{i=1}^3 \exp(-y \, a_{\n,i}(\omega)) \, v_i \, \w_i(\omega,\n)$$

A different expansion can be obtained using the last line of \eqref{eq:OUN}. Expand 
the amplitude vector $\aA$ in terms of three independent eigenvectors $\w^1_j(\omega)$  of the matrix
$\C_\n(\omega)$ corresponding to the eigenvalues $1/c_j(\omega)$. This yields an
expansion of the plane wave in terms of quasi-elastic modes \index{quasi-elastic modes} 
$$\exp(y [\ii \omega \C_\n(\omega) - \A_\n(\omega)]) \, \aA = \Psi(y, \omega)\, \exp( 
-y\, \A_\n(\omega)) \sum_{j=1}^3 \exp( \ii \omega \, y/c_{\n,j}(\omega)) \, v^1_j \, \w^1_j(\omega)$$
The quasi-elastic modes appearing under the sum are coupled by the attenuation operator 
$\exp(-\A_\n(\omega)\, y)$ and by the factor $\Psi(y,\omega)$. 

The functions $\Phi$ and $\Psi$ are rather difficult to evaluate numerically, but see \cite{CasasAl} concerning the numerical implementation of the Zassenhaus formula. 
 
\section{Energy flux density of an inhomogeneous plane wave.}
\label{sec:flux}

Equation~\eqref{eq:homoeqmotion} implies the energy conservation equation
$$\frac{\dd}{\dd t} \left[\u_{,t}^2 + W\right] + \Div \Psi = 0$$
where the energy flux density is given by the formula 
\begin{equation} \label{eq:flux2}
\Psi_l = -\sigma_{kl}\, \dot{u}_k  =
-\dot{u}_k \, G_{klmn}\ast \dot{u}_{m,n}, \qquad l=1,2,3
\end{equation}
and $\dd W/\dd t = \sigma_{kl}\, \dd \u_{k,l}/\dd t$. The energy density functional 
$W$ for LICM relaxation was constructed in \cite{HanEnergy} and for CPD relaxation in 
\cite{HanHamiltonianVE}.

We shall now calculate the energy flux density $\mathbf{\Psi}$ of an inhomogeneous 
plane wave 
$\u = \Re\left[\aA \, \exp(-\ii \omega t + 
\ii \, \kk \cdot \x)\right]$, with a complex amplitude vector 
$\aA = \aA^\mathrm{R} + \ii \aA^\mathrm{I}$ and a complex wave number
$\kk = \kk^\mathrm{R} + \ii \kk^\mathrm{I}$, and then average the flux over 
the period 
$T = 2 \uppi/\omega$. In contrast to the previous sections we do not assume here that
the vectors $\kk^{\mathrm{R}}$ and  $\kk^{\mathrm{I}}$ are collinear. This concept is a generalization 
of the plane waves discussed in Section~\ref{sec:plwav}.

The time average is denoted by
$$\langle \Psi_l \rangle := \frac{1}{T} \int_0^T \Psi_l(t,x) \, \dd t$$
Let $C = \cos\left(\omega t - \kk^\mathrm{R}\cdot\x\right)$,
$S = \sin\left(\omega t - \kk^\mathrm{R}\cdot\x\right)$.
The particle velocity and the stress are  given by the expressions
$$\dot{\u} = \omega \,
\e^{-\kk^\mathrm{I}\cdot\x}\, \left[\aA^\mathrm{I}\,C - \aA^\mathrm{R}\, 
S\right]$$
and
\begin{multline}
\sigma_{kl} = G_{klmn}\ast\dot{u}_{m,n}  = \\
\omega \, G_{klmn}\ast \left[\Re(a_m \, k_n)\, C + \im(a_m \, k_n)\,S
   \right]\, \e^{-\kk^\mathrm{I}\cdot\x}= \\
\omega \, G_{klmn}\ast \, \Re\left[ a_m \, k_n \, \e^{-\ii \omega t 
+ \ii \kk^\mathrm{R}\cdot\x} \right] \, \e^{-\kk^\mathrm{I}\cdot\x}  
\end{multline}
The identity
$$\int_0^\infty \tens{G}(s) \, \e^{-\ii \omega (t-s)} \, \dd s 
= \e^{-\ii \omega t}\, \tilde{\tens{G}}(-\ii \omega)$$
implies that
$$\sigma_{kl} = \omega\, \e^{-\kk^\mathrm{I}\cdot\x}\,\re\left[a_m \,
k_n 
\tilde{G}_{klmn}(-\ii \omega)  \, \e^{-\ii \omega t + 
\ii \kk^\mathrm{R}\cdot\x}   \right]$$
Hence, noting that $\langle C^2 \rangle = \langle S^2 \rangle = 1/2$ 
and $\langle S C \rangle = 0$, 
\begin{multline}
-\langle \Psi_l \rangle = \langle \dot{u}_k \, \sigma_{kl}\rangle 
= \omega^2 \, \e^{-2 \kk^\mathrm{I}\cdot\x}\,\\ \times \left\langle 
\left[-a^\mathrm{R}\, S + a^\mathrm{I}\, C\right] \,
\left[\Re \left[ a_m \, k_n \, \tilde{G}_{klmn}(-\ii\omega)
 \right]\, C + \Im \left[ a_m \, k_n \, \tilde{G}_{klmn}(-\ii\omega)
 \right]\, S  \right] \right\rangle
\end{multline}
and 
\begin{equation} \label{eq:flux3}
-\langle \Psi_l \rangle  = -\frac{\omega^2}{2} \e^{-2 \kk^\mathrm{I}\cdot\x}\, 
\im\left[ \overline{a_k}\, \tilde{G}_{klmn}(-\ii\omega)  
a_m \, k_n\right], \qquad l=1,2,3
\end{equation}
For comparison with elastic media it is more convenient to express the energy flux density
in terms of $\tens{Q}(p) = p \, \tilde{\tens{G}}(p)$:
\begin{equation} \label{eq:flux4}
-\langle \Psi_l \rangle  = \frac{\omega}{2} \e^{-2 \kk^\mathrm{I}\cdot\x}\, 
\re\left[ \overline{a_k}\, \tilde{Q}_{klmn}(-\ii\omega)  
a_m \, k_n\right], \qquad l=1,2,3
\end{equation}
For an elastic medium $\tilde{Q}$ is constant (the stiffness tensor).

\section{CPD functions and the direction of the energy flux density.}
\label{sec:CPD}

The energy flux density decays exponentially in the direction of 
the attenuation vector. We shall now show that the energy flux density decays 
exponentially in the direction of the flux. To this effect we shall use a much weaker 
assumption about the relaxation modulus than previously.

\begin{definition} \label{def:CPD}
A locally integrable function $f: \mathbb{R}_+ \rightarrow \mathbb{C}$ is 
said to be {\em causal positive definite} (CPD) if 
\begin{equation} \label{eq:PT+}
\Re \int_0^\infty f(t) \, (\phi\ast\check{\phi})(t) \dd t 
\equiv \re \int_{-\infty}^\infty \phi(t) \int_0^\infty f(s) \,
\overline{\phi(t-s)} \, \dd s \, \dd t \geq 0
\end{equation} 
for every square-integrable function $\phi$ with compact support.
\end{definition}
We have used here the notation $\check{f}(t) := \overline{f(-t)}$.
In \cite{GripenbergLondenStaffans} CPD functions are called functions of positive type.

This definition is readily generalized to matrix- and operator-valued functions. In 
particular we are interested in tensor-valued functions. In this context tensors are defined as 
operators on the space $S$ of symmetric $d\times d$ complex matrices. The space $S$ is 
endowed with the scalar product $\langle \ee, \f\rangle := \sum_{k,l} e_{kl}\, f_{kl}$.
The space of symmetric operators on $S$ will be denoted by $C$. 
\begin{definition} \label{def:tCPD}
A locally integrable function tensor-valued function $\tens{G}: \mathbb{R}_+ \rightarrow C$ is 
said to be {\em causal positive definite} (CPD) if 
\begin{equation} \label{eq:PT+}
\Re \int_0^\infty G_{ijkl}(t) (e_{ij}\ast\check{e}_{kl})(t) \dd t 
\equiv \Re \int_{-\infty}^\infty \left[ \int_0^\infty \langle \ee(t),  \tens{G}(s) \,
\overline{\ee(t-s)} \rangle\, \dd s \right] \, \dd t \geq 0
\end{equation} 
for every $S$-valued square-integrable function $\ee$ on $\mathbb{R}$ with compact support.
\end{definition}

We shall need spectral characterizations of CPD functions.
\begin{theorem} \label{thm:BochnerGeneralized1}
$\tens{G}$ is a tensor-valued CPD function on $\mathbb{R}_+$ if the function
$$\tens{F}(t) = \begin{cases} \tens{G}(t) & t > 0 \\
0 & t \leq 0 \end{cases} 
$$
is a tempered distribution and the real part of its Fourier transform 
$\tens{N}(\xi) := \Re \hat{\tens{F}}(\xi)$ is positive semi-definite for $\omega \in \mathbb{R}$. 
\end{theorem}
\noindent \cite{GripenbergLondenStaffans}, Theorem~16.2.5.

Since $\langle \ee, \tens{N}(\xi) \,\ee \rangle$ is non-negative for every $\ee \in S$ (in the distributions sense), it is a positive Radon measure. The polarization argument (see the proof of Theorem~\ref{thm:CBF})
leads to the conclusion that $\tens{N}$ is a tensor-valued Radon measure.

Let $\varphi_n(\xi), n=1,2\ldots$,  be a sequence of real Schwartz test functions tending in the 
$\mathcal{C}^0$ norm to the characteristic function $\chi_E$ of a Borel set $E \subset \mathbb{R}$. 
Since 
\begin{multline*}
\int \varphi_n(\xi) \, \tens{N}(\xi) \, \dd \xi = \frac{1}{2} \int \left[ \hat{\tens{F}}(\xi) +
 \hat{\tens{F}}(\xi)^\dag \right] \, \varphi_n(\xi) \, \dd \xi = \\
\frac{1}{2} \int \hat{\tens{F}}(\xi) \, \varphi_n(\xi) \, \dd \xi + \frac{1}{2} \left[ \int \hat{\tens{F}}(\xi) \, \varphi_n(\xi)\, \dd \xi \right]^\dag
\end{multline*}
 is Hermitian, so is is limit 
$$\tens{N}(E) = \lim_{n\rightarrow \infty} \int \varphi_n(\xi)\, \tens{N}(\dd \xi)$$
for every Borel $E \subset \mathbb{R}$. 

Inverting the Fourier transform and noting that $\tens{F}(t) = 0$ for $t < 0$ 
while $\tens{F}(t)$ is real and $\tens{F}(t) = \tens{F}(t)^\top$ ($F_{ijkl}(t) = F_{klij}(t)$), we have
\begin{multline} \label{eq:BochnerGen}
\tens{G}(t) = \tens{F}(t) + \overline{\tens{F}(-t)} = \frac{1}{2 \uppi} \int_{-\infty}^\infty \e^{-\ii \xi t} \, \left[ \tens{H}(\dd \xi) + \tens{H}(\dd \xi)^\dag\right] = \frac{1}{\uppi} \int_{-\infty}^\infty \e^{-\ii \xi t} \, \tens{N}(\dd \xi)  
\end{multline}
for $t > 0$. 

Using a version of inequality \eqref{eq:RN} one simplify the above results:
\begin{theorem} \label{thm:scalmeas}
If the Radon measure $n$ is the trace of $\tens{N}$, 
$\int \phi(\xi)\, n(\dd \xi) := \sum_{k,l=1}^3 \int \phi(\xi) N_{klkl}(\dd \xi)$ for 
every $\phi \in \mathcal{C}^0_\mathrm{c}$, then
$n \geq 0$ and and there is an essentially bounded $C$-valued function 
$\tens{L}: \mathbb{R}_+  \rightarrow C$ such that  $\tens{N}(\dd \xi) =
\tens{L}(\xi) \, n(\dd \xi)$, $\vert \tens{L}(\xi) \vert 
\leq 1$ for $n$-almost all $\xi \in \mathbb{R}$.
\end{theorem}
\noindent where $\vert \tens{T} \vert^2 := \sum_{ijkl} T_{ijkl}^{\;2}$ for every $\tens{T} \in C$.
We are now ready to state the main result of this section:
\begin{theorem}
Let $\tens{G}$ be a tensor-valued CPD function.

The angle between the time-averaged energy flux density $\mathbf{\Psi}$ and the 
attenuation vector $\kk^{\mathrm{I}}$ is acute.
\end{theorem}
\begin{proof}
The Fourier transform $\hat{\tens{G}}(\omega) \equiv \tilde{\tens{G}}(-\ii \omega)$ 
of a CPD function $\tens{G}$ satisfies the inequality $\Re \hat{\tens{G}}(\omega) \geq 0$,
hence 
\begin{multline} \label{eq:aa}
\Re\left[ \overline{a_k} \,  \overline{k_l} \, \tilde{G}_{klmn}(-\ii \omega) 
\, a_m \, k_n \right] \equiv \\ \Re\left[ \overline{a_k} \,  
\tilde{\tens{G}}(-\ii \omega) \, a_m \, k_n \right] \, k^\mathrm{R}_{\;l} 
+\Im \left[ \overline{a_k} \,  
\tilde{\tens{G}}(-\ii \omega) \, a_m \, k_n \right] \, k^\mathrm{I}_{\;l}  \geq 0
\end{multline}
On the other hand
$$\tilde{G}_{klmn}(-\ii \omega)\, k_l\, k_n \, a_m = -\ii \omega\, \rho \,a_k \ $$
Hence
\begin{multline} \label{eq:bb}
\Re \left[ \overline{a_k} \,  \overline{k_l} \, \tilde{G}_{klmn}(-\ii \omega) 
\, a_m \, k_n \right] \equiv \\ \Re \left[ \overline{a_k} \, \tilde{G}_{klmn}(-\ii \omega) 
\, a_m \, k_n \right]\, k^\mathrm{R}_{\;l} -\Im \left[ \overline{a_k} \, 
\tilde{G}_{klmn}(-\ii \omega) 
\, a_m \, k_n \right]\, k^\mathrm{I}_{\;l} = 0
\end{multline}
Subtracting equation~\eqref{eq:bb} from equation~\eqref{eq:aa} yields the inequality
$$2 \Im \left[ \overline{a_k} \, \tilde{G}_{klmn}(-\ii \omega) 
\, a_m \, k_n \right]\, k^\mathrm{I}_{\;l} \geq 0$$
which on account of \eqref{eq:flux3},  implies that 
\begin{equation}
\langle \Psi \rangle \cdot \kk^\mathrm{I} \geq 0
\end{equation}
q.e.d.
\end{proof}

A relation with the LICM relaxation moduli can be established by means of the following theorem 
\begin{theorem} \label{thm:vx}
A locally integrable $d\times d$ matrix-valued function $\A(t)$ on $\mathbb{R}_+$ satisfying the 
condition that for each $\vv \in \mathbb{R}^d$ the function
$\vv^\dag \, \A(t)\, \vv $ is non-negative, non-increasing and convex, is CPD. 
\end{theorem}
\noindent cf \cite{GripenbergLondenStaffans}, Theorem~16.3.1. 

Theorem~\ref{thm:vx} can be readily generalized to operator-valued functions 
such as the tensor-valued functions. In particular every tensor-valued LICM 
function satisfies the hypotheses of Theorem~\ref{thm:vx} is CPD. 
\begin{corollary} \label{cor:ux}
If the relaxation modulus is LICM then the angle between the time-averaged energy flux density and 
the attenuation vector is acute.
\end{corollary} 

The LICM property is defined in terms of an infinite sequence of inequalities \eqref{eq:CM} for derivatives 
of all the integer orders. However only the first three of them are relevant for the thesis of 
Corollary~\ref{cor:ux}. This is interesting because the CM property cannot be verified from raw
data while the first three inequalities \eqref{eq:CM} can.

Viscoelastic media with CPD relaxation moduli were studied in \cite{HanHamiltonianVE}.

%\bibliography{ownnew12,mathnew12}
\end{document}